\documentclass[11pt,fleqn]{article}

\usepackage{amsmath,amssymb,amsthm}%,cite,backref

\newcommand{\noprint}[1]{}

\newcommand{\diag}{\mathop{\rm diag}\nolimits}

\newtheorem{theorem}{Theorem}
\newtheorem{lemma}{Lemma}
\newtheorem{corollary}{Corollary}

{\theoremstyle{definition} \newtheorem{definition}{Definition}

\newtheorem{note}{Note}

\begin{document}

\par\noindent {\LARGE\bf
Equivalence of Diagonal Contractions to Generalized \\ IW-Contractions with Integer Exponents\par}

{\vspace{4mm}\par\noindent 
Dmytro R. POPOVYCH~$^\dag$ and Roman O. POPOVYCH~$^\ddag$
\par\vspace{2mm}\par}

{\vspace{2mm}\par\it
\noindent $^\dag$Faculty of Mechanics and Mathematics, National Taras Shevchenko University of Kyiv,\\
$\phantom{^{\dag}}$building 7, 2, Academician Glushkov prospectus, Kyiv, Ukraine, 03127
\par}

{\vspace{2mm}\par\it
\noindent $^\ddag$Fakult\"at f\"ur Mathematik, Universit\"at Wien, Nordbergstra{\ss}e 15, A-1090 Wien, Austria\\
$\phantom{^{\ddag}}$Institute of Mathematics of NAS of Ukraine, 3 Tereshchenkivska Str., Kyiv-4, Ukraine
 \par}

{\vspace{2mm}\par\noindent\rm E-mail: \it  deviuss@gmail.com, rop@imath.kiev.ua
 \par}

{\vspace{7mm}\par\noindent\hspace*{5mm}\parbox{150mm}{\small
We present a simple and rigorous proof of the claim by Weimar-Woods [{\it Rev. Math. Phys.}, 2000, {\bf 12}, 1505--1529]
that any diagonal contraction (e.g., a generalized In\"on\"u--Wigner contraction) is equivalent to
a generalized In\"on\"u--Wigner contraction with integer parameter powers.
}\par\vspace{5mm}}

\noprint{
% Equivalence of diagonal contractions to generalized IW-contractions with integer exponents
% Dmytro R. Popovych, Roman O. Popovych

\noindent 
{\bf Keywords:} \
contraction of Lie algebras, generalized IW-contraction, 
diagonal contraction,
degeneration of Lie algebras, one-parametric subgroup degeneration

\vspace{2mm}

\noindent 
{\bf MSC:} \ 17B81; 17B70
% PACS numbers
}

\section{Introduction}

Usual or generalized In\"on\"u--Wigner contractions (\emph{IW-contractions}) is a conventional way
for realizing contractions of Lie algebras.
The most known examples on contractions of Lie algebras arising in physics
(contractions from the Poincar\'e algebra to the Galilean one and
from the Heisenberg algebras to the Abelian ones of the same dimensions,
forming a symmetry background of limit processes from relativistic and quantum mechanics to classical mechanics)
are represented by usual IW-contractions.
The second of the above examples is a trivial contraction.
Any Lie algebra is contracted to the Abelian algebra of the same dimension via the IW-contraction
corresponding to the zero subalgebra.

The concept of contractions of Lie algebras introduced by Segal~\cite{Segal1951} in a heuristic way
became well known only after the invention of IW-contractions in~\cite{Inonu&Wigner1953, Inonu&Wigner1954}.
Saletan~\cite{Saletan1961} gave the first rigorous general definition of contractions and
investigated the whole class of one-parametric contractions
whose matrices are first-order polynomials with respect to contraction parameters. 
Later contractions of Lie algebras appeared in different areas of physics and mathematics,
e.g., in the study of representations, invariants and special functions.

The name ``generalized In\"on\"u--Wigner contraction'' were first used in~\cite{Hegerfeldt1967}
for so-called $p$-contractions by Doebner and Melsheimer \cite{Doebner&Melsheimer1967}.
Generalizing IW-contractions, they studied contractions 
whose matrices become diagonal after choosing appropriate bases of initial and contracted algebras, 
and diagonal elements being powers of the contraction parameter with real exponents.
In the algebraic literature, similar contractions with integer exponents are called \emph{one-parametric subgroup degenerations}
\cite{Burde1999,Burde2005,Burde&Steinhoff1999,Grunewald&Halloran1988}.
The notion of degenerations of Lie algebras extends the notion of contractions to the case of an arbitrary algebraically closed field
and is defined in terms of orbit closures with respect to the Zariski topology.
Note that in fact a one-parametric subgroup degeneration is induced by a one-parametric matrix group only
under an agreed choice of bases in the corresponding initial and contracted algebras.

All continuous contractions arising in the physical literature are generalized IW ones.
The question whether every contraction can be realized by a generalized IW-contraction was posed in~\cite{Weimar-Woods1991}.
Later it was conjectured that the answer is positive \cite{Weimar-Woods2000}.
The attempt of proving the conjecture in~\cite{Weimar-Woods2000} was not successful since,
as shown in~\cite{Nesterenko&Popovych2006}, the proof contains an unavoidable incorrectness at the initial step.
In fact, contrary instances on this conjecture, involving characteristically nilpotent Lie algebras,
were earlier constructed by Burde \cite{Burde1999,Burde2005} for all dimensions not less than seven
but they are not well-known among the physical community.
Since each proper generalized IW-contraction induces a proper grading on the contracted algebra and
each characteristically nilpotent Lie algebra possesses only nilpotent derivations and hence has no proper gradings
then \emph{any contraction to characteristically nilpotent Lie algebras obviously is inequivalent to a generalized IW-contraction}.
This fact cannot be used for lower dimensions in view of the absence of characteristically nilpotent Lie algebras in dimensions less than seven.

Examples of another kind on non-universality of generalized IW-contractions were recently presented in~\cite{Popovych&Popovych2009}.
There exist exactly two (resp.\ one) well-defined contractions of four-dimensional Lie algebras over the field of real (resp.\ complex) numbers,
which are inequivalent to generalized IW-contractions in spite of that the contracted algebras possess a wide range of proper gradings.
These examples are important since they establish new bounds for applicability of generalized IW-contractions, 
showing that a well-defined contraction may be inequivalent to a generalized IW-contraction 
even if the contracted algebra possesses proper gradings. 
The other contractions of four-dimensional Lie algebras were realized in \cite{Campoamor-Stursberg2008,Nesterenko2008,Nesterenko&Popovych2006}
by generalized IW-contractions involving nonnegative integer parameter exponents not greater than three,
and the upper bound proved to be exact \cite{Popovych&Popovych2009}.
Uniting these results gives the exhaustive description of generalized IW-contractions in dimension four.
It is expected that a similar answer may be true for dimensions five and six.
Therefore, generalized IW-contractions seem to be universal realizations only for contractions of Lie algebras of dimensions not greater than three
\cite{Nesterenko&Popovych2006}.

Considering different subclasses of Lie algebras closed with respect to contractions
or imposing restrictions on contraction matrices, we can pose the problem on partial universality of generalized IW-contractions
in specific subsets of contractions.
Thus, generalized IW-contractions of nilpotent algebras are studied in~\cite{Burde&Nesterenko&Popovych2009}.
Diagonal contractions arose in~\cite{Weimar-Woods2000} as an intermediate step in realizing general contractions
via generalized IW-contractions.
It was indirectly claimed as a part of a more general incorrect theorem on universality of generalized IW-contractions
that every diagonal contraction is equivalent to a generalized IW-contraction
and every generalized IW-contraction is equivalent to a generalized IW-contraction with integer exponents.
Although the claim is correct and important for the theory of contractions, 
the initial step of the proof presented in~\cite{Weimar-Woods2000} has to be corrected to avoid an essential inconvenience 
(especially for the case of complex Lie algebras)
partially induced by incorrectness of preliminary results and the general theorem. 
%(See Section~\ref{SectionDiscussion} for discussion of this.)

In this paper we rigorously prove two statements.
The proof that \emph{integer exponents are sufficient for generalized IW-contractions} is rather geometrical.
The second theorem which states \emph{equivalence of every diagonal contraction to
a generalized In\"on\"u--Wigner contraction involving integer powers of a parameter}
is proved in a more algebraic way.
The first statement obviously follows from the second one but it is of independent interest
and hence is separately formulated and proved.
In particular, it connects the investigations of generalized IW-contractions (admitting real exponents)
in the physical literature and one-parametric subgroup degenerations (whose parameter exponents are necessarily integer)
in the algebraic literature. 
The proofs are essentially simpler than those from~\cite{Weimar-Woods2000} and so much algorithmic 
that the described algorithms can be directly realized in symbolic calculation programs. 
This completely solves the problem on universality of generalized IW-contractions in the set of diagonal contractions.

\section{Contractions and generalized IW-contractions}
\label{SectionOnContractionsAndGenIWContractions}

The notion of contraction is defined for arbitrary algebraically closed fields in terms of orbit closures in the variety of Lie algebras
\cite{Burde1999,Burde2005,Burde&Steinhoff1999,Grunewald&Halloran1988,Lauret2003}.
Let $V$ be an $n$-dimensional vector space over an algebraically closed field~$\mathbb F$, $n<\infty$, and
$\mathcal L_n=\mathcal L_n(\mathbb F)$ denote the set of all possible Lie brackets on~$V$.
We identify $\mu\in\mathcal L_n$ with the corresponding Lie algebra $\mathfrak g=(V,\mu)$.
$\mathcal L_n$ is an algebraic subset of the variety $V^*\otimes V^*\otimes V$ of bilinear maps from $V\times V$ to $V$.
Indeed, under setting a basis $\{e_1,\dots,e_n\}$ of~$V$
there is the one-to-one correspondence between $\mathcal L_n$ and
\[
\mathcal C_n=\{(c_{ij}^k)\in\mathbb F^{n^3}\mid c_{ij}^k+c_{ji}^k=0,\,
c_{ij}^{i'\!}c_{i'\!k}^{k'\!}+c_{ki}^{i'\!}c_{i'\!j}^{k'\!}+c_{jk}^{i'\!}c_{i'\!i}^{k'\!}=0\},
\]
which is determined for any Lie bracket $\mu\in\mathcal L_n$ and
any structure constant tuple $(c_{ij}^k)\in\mathcal C_n$ by the formula $\mu(e_i,e_j)=c_{ij}^ke_k$.
Throughout the indices $i$, $j$, $k$, $i'$, $j'$ and $k'$ run from 1 to $n$
and the summation convention over repeated indices is used.
$\mathcal L_n$ is called the \emph{variety of $n$-dimensional Lie algebras (over the field~$\mathbb F$)}
or, more precisely, the variety of possible Lie brackets on~$V$.
The group~${\rm GL}(V)$ acts on $\mathcal L_n$ in the following way:
\[
(U\cdot\mu)(x,y)=U^{-1}\bigl(\mu(Ux,Uy)\bigr)\quad \forall U\in {\rm GL}(V),\forall \mu\in\mathcal L_n,\forall x,y\in V.
\]
(This is the right action conventional for the `physical' contraction theory.
In the algebraic literature, the left action defined by the formula $(U\cdot\mu)(x,y)=U\bigl(\mu(U^{-1}x,U^{-1}y)\bigr)$
is used that is not of fundamental importance.)
Denote the orbit of $\mu\in\mathcal L_n$ under the action of~${\rm GL}(V)$ by $\mathcal O(\mu)$ and
the closure of it with respect to the Zariski topology on~$\mathcal L_n$ by $\overline{\mathcal O(\mu)}$.

\begin{definition}\label{DefOfContractionsViaOrbitClosure}
The Lie algebra $\mathfrak g_0=(V,\mu_0)$ is called a \emph{contraction} (or \emph{degeneration})
of the Lie algebra~$\mathfrak g=(V,\mu)$ if $\mu_0\in\overline{\mathcal O(\mu)}$.
The contraction is \emph{proper} if $\mu_0\in\overline{\mathcal O(\mu)}\backslash\mathcal O(\mu)$.
The contraction is \emph{nontrivial} if $\mu_0\not\equiv0$.
\end{definition}

In the case $\mathbb F=\mathbb C$ the orbit closures with respect to the Zariski topology coincide
with the orbit closures with respect to the Euclidean topology and Definition~\ref{DefOfContractionsViaOrbitClosure} is
reduced to the usual definition of contractions which is also suitable for the case $\mathbb F=\mathbb R$.

\begin{definition}\label{DefOfContractions1}
Consider a parameterized family of the Lie algebra $\mathfrak g_\varepsilon=(V,\mu_\varepsilon)$ isomorphic to $\mathfrak g=(V,\mu)$.
The family of the new Lie brackets~$\mu_\varepsilon$, $\varepsilon \in (0,1]$, is defined via the Lie bracket~$\mu$
with a continuous function $U\colon (0,1]\to {\rm GL}(V)$ by the rule
$\mu_\varepsilon(x,y)=U_\varepsilon{}^{-1}\mu(U_\varepsilon x,U_\varepsilon y)$ $\forall \; x, y\in V$.
If for any $x, y\in V$ there exists the limit
\[
\lim\limits_{\varepsilon \to +0}\mu_\varepsilon(x,y)=
\lim\limits_{\varepsilon \to +0}U_\varepsilon{}^{-1}\mu(U_\varepsilon x,U_\varepsilon y)=:\mu_0(x,y)
\]
then $\mu_0$ is a well-defined Lie bracket.
The Lie algebra $\mathfrak g_0=(V,\mu_0)$ is called a \emph{one-parametric continuous contraction}
(or simply a \emph{contraction}) of the Lie algebra~$\mathfrak g$.
The procedure $\mathfrak g\to\mathfrak g_0$
providing $\mathfrak g_0$ from~$\mathfrak g$ is also called a \emph{contraction}.
\end{definition}

If a basis of~$V$ is fixed, the operator $U_\varepsilon$ is defined by the corresponding matrix $U_\varepsilon\in{\rm GL}_n(\mathbb F)$
and Definition~\ref{DefOfContractions1} can be reformulated in terms of structure constants.
Let ${c}^k_{ij}$ be the structure constants of the algebra~$\mathfrak g$ in the fixed basis~$\{e_1, \ldots, e_n\}$.
Then Definition~\ref{DefOfContractions1} is equivalent to that the limit
\[\lim\limits_{\varepsilon\to+0}(U_\varepsilon)_{i'}^i(U_\varepsilon)_{j'}^j(U_\varepsilon{}^{-1})_k^{k'}c^{k}_{ij}=:c^{k'}_{0,i'\!j'}\]
exists for all values of $i'$, $j'$ and $k'$ and, therefore,
$c^{k'}_{0,i'\!j'}$ are components of the well-defined structure constant tensor of a Lie algebra~$\mathfrak g_0$.
The parameter $\varepsilon$ and the matrix-function $U_\varepsilon$ are called a \emph{contraction parameter} and a \emph{contraction matrix},
respectively.

The contraction $\mathfrak g\to\mathfrak g_0$ is called
\emph{trivial} if ${\mathfrak g}_0$ is Abelian and \emph{improper} if ${\mathfrak g}_0$ is isomorphic to ${\mathfrak g}$.

\begin{definition}\label{equivalent contr}
The contractions $\mathfrak g\to\mathfrak g_0$ and $\tilde{\mathfrak g}\to\tilde{\mathfrak g}_0$ are called \emph{(weakly) equivalent}
if the algebras $\tilde{\mathfrak g}$ and $\tilde{\mathfrak g}_0$ are isomorphic to $\mathfrak g$ and $\mathfrak g_0$, respectively.
\end{definition}

Using the weak equivalence concentrates one's attention on existence and results of contractions and
neglects differences in ways of contractions.
To take into account such ways, we can introduce different notions of stronger equivalence~\cite{Nesterenko&Popovych2006}.

The following useful lemma is obvious. 

\begin{lemma}
If the matrix~$U_\varepsilon$ of a contraction $\mathfrak g\to\mathfrak g_0$ can be 
represented in the form $U_\varepsilon=\hat U_\varepsilon\check U_\varepsilon$, where 
$\hat U$ and $\check U$ are continuous functions from $(0,1]$ to ${\rm GL}_n(\mathbb F)$ and 
$\exists\lim_{\varepsilon \to +0}\check U_\varepsilon=:\check U_0\in{\rm GL}_n(\mathbb F)$
%the limit $\check U_0$ of $\check U$ at $\varepsilon \to +0$ exists and is a nonsingular matrix 
then $\hat U_\varepsilon\check U_0$ also is a matrix of the contraction $\mathfrak g\to\mathfrak g_0$ 
and the matrix~$\hat U_\varepsilon$ leads to an equivalent contraction.
\end{lemma}

Generalized In\"on\"u--Wigner contractions is defined as a specific way for realizations of general contractions.

\begin{definition}\label{DefOfGenIWContractions}
The contraction $\mathfrak g\to\mathfrak g_0$ (over $\mathbb C$ or $\mathbb R$)
is called \emph{a generalized In\"on\"u--Wigner contraction} if
its matrix $U_\varepsilon$ can be represented in the form
$U_\varepsilon=AW_\varepsilon P$, where $A$ and $P$ are constant nonsingular matrices
and $W_\varepsilon=\diag(\varepsilon^{\alpha_1},\dots,\varepsilon^{\alpha_n})$ for some $\alpha_1,\dots,\alpha_n\in\mathbb R$.
The tuple of exponents $(\alpha_1,\dots,\alpha_n)$ is called the \emph{signature} of the generalized IW-contraction $\mathfrak g\to\mathfrak g_0$.
\end{definition}

In fact, the signature of a generalized IW-contraction $\mathfrak C$ is defined up to a positive multiplier
since the reparametrization $\varepsilon=\tilde\varepsilon^\beta$, where $\beta>0$,
leads to a generalized IW-contraction strongly equivalent to $\mathfrak C$.

Due to the possibility of changing bases in the initial and contracted algebras, we can set $A$ and $P$ equal to the unit matrix.
This is appropriate for some theoretical considerations but not for working with specific Lie algebras.
For $U_\varepsilon=\diag(\varepsilon^{\alpha_1},\dots,\varepsilon^{\alpha_n})$
the structure constants of the resulting algebra $\mathfrak g_0$ are calculated by the formula
$c_{0,ij}^k=\lim_{\varepsilon\to +0}c_{ij}^k\,\varepsilon^{\alpha_i+\alpha_j-\alpha_k}$
with no summation with respect to the repeated indices.
Therefore, the constraints
$\alpha_i+\alpha_j\geqslant \alpha_k$ if $c_{ij}^k\ne 0$
are necessary and sufficient
for the existence of the well-defined generalized IW-contraction with the contraction matrix $U_\varepsilon$,
and $c_{0,ij}^k=c_{ij}^k$ if $\alpha_i+\alpha_j=\alpha_k$ and $c_{0,ij}^k=0$ otherwise.
This obviously implies that the conditions of existence of generalized IW-contractions and the structure of contracted algebras
can be reformulated in the basis-independent terms of gradings of contracted algebras associated with
filtrations on initial algebras~\cite{Grunewald&Halloran1988,Lyhmus1969}.
(Probably, this was a motivation for introducing and developing the purely algebraic notion of \emph{graded contractions}
\cite{Hrivnak&Novotny&Patera&Tolar2006,Montigny&Patera1991}.)
In particular, the contracted algebra~$\mathfrak g_0$ has to admit a derivation
whose matrix is diagonalizable to $\diag(\alpha_1,\dots,\alpha_n)$.

The following statement is known as a conjecture for a long time (see, e.g., \cite{Weimar-Woods2000}).

\begin{theorem}\label{TheoremOnGenIWContractions}
Any generalized IW-contraction is equivalent to
a generalized IW-contraction with an integer signature (and the same associated constant matrices).
\end{theorem}

\begin{proof}
Let the contraction $\mathfrak g\to\mathfrak g_0$ have the matrix $U_\varepsilon=AW_\varepsilon P$,
where $A$ and $P$ are constant nonsingular matrices
and $W_\varepsilon=\diag(\varepsilon^{\alpha_1},\dots,\varepsilon^{\alpha_n})$
for some $\alpha=(\alpha_1,\dots,\alpha_n)\in\mathbb R^n$.
We introduce the notations
\[
\mathcal E=\{(i,j,k)\mid c_{ij}^k\ne 0,\,\alpha_i+\alpha_j=\alpha_k\},\quad
\mathcal N=\{(i,j,k)\mid c_{ij}^k\ne 0,\,\alpha_i+\alpha_j>\alpha_k\}.
\]
We can assume that $\mathcal N\ne\varnothing$ since otherwise the contraction $\mathfrak g\to\mathfrak g_0$ is improper 
and therefore, equivalent to a generalized IW-contraction with the zero signature. 
The system $S$ of the equations $x_i+x_j=x_k$, $(i,j,k)\in\mathcal E$, and
the inequalities $x_i+x_j>x_k$, $(i,j,k)\in\mathcal N$, for $x=(x_1,\dots,x_n)\in\mathbb R^n$
is compatible since it has the solution $x=\alpha$.
If $x$'s satisfy $S$ then $\lambda x$ is a solution of~$S$ for any positive real $\lambda$.
Therefore, the solution set of $S$ is a nonempty conus in~$\mathbb R^n$.
We express a maximal subset of $x$'s via the other $x$'s using the equations $x_i+x_j=x_k$, $(i,j,k)\in\mathcal E$.
Denote by $I$ (resp.~$\bar I$) the set of numbers of the expressed $x$'s (resp.\ unconstrained $x$'s);
$\bar I$ is complimentary to~$I$ in $\{1,\dots,n\}$.
The expressions for $x_i$, $i\in I$, are linear in $x_j$, $j\in\bar I$, and have rational coefficients.
After substituting these expressions into the inequalities $x_i+x_j>x_k$, $(i,j,k)\in\mathcal N$,
we obtained a system~$S'$ of strict inequalities for $x_j$, $j\in\bar I$,
which defines an open nonempty conus in $\mathbb R^{|\bar I|}$.
The conus necessarily contains points with rational coordinates.
This means that the system~$S$ possesses rational solutions. Therefore, it also has integer solutions.
Any solution of~$S$ is the signature of a well-defined generalized IW-contraction $\mathfrak g\to\mathfrak g_0$
with the same associated constant matrices~$A$ and~$P$.
\end{proof}

Theorem~\ref{TheoremOnGenIWContractions} is a consequence of a more general statement on diagonal contractions,
discussed in the next section.

\begin{note}
In fact, the proof of Theorem~\ref{TheoremOnGenIWContractions} gives a \emph{constructive way} for finding 
an integer signature via solving the system~$S$, e.g., by the Gaussian elimination \cite{Solodovnikov1977}. 
(See also \cite{Chernikov1968} for different methods of solving linear systems of inequalities.)
At first the system~$S$ is reduced to the system~$S'$ by the Gaussian elimination of $x_i$, $i\in I$, 
due to the equations $x_i+x_j=x_k$, $(i,j,k)\in\mathcal E$. 
Then the Gaussian elimination is applied to the (compatible) system~$S'$ of strict inequalities. 
After the final step of the elimination we take rational (e.g., zero) values for the residual components of~$x$ 
and go back with the elimination conditions step-by-step, choosing rational values for the steps 
when the corresponding components of~$x$ are constrained by inequalities.  

The similar remark is true for the proof of Theorem~\ref{TheoremOnDiagonalContractions}. 
\end{note}

The notion of sequential contractions is introduced similarly to continuous contractions~\cite{Segal1951,Weimar-Woods2000}.
Namely, consider a sequence of the Lie algebra $\mathfrak g_p=(V,\mu_p)$, $p\in\mathbb N$, isomorphic to $\mathfrak g=(V,\mu)$.
The sequence of the new Lie brackets~$\{\mu_p,\,p\in\mathbb N\}$, is defined via the Lie bracket~$\mu$
with a sequence $\{U_p,\,p\in\mathbb N\}\subset{\rm GL}(V)$ by the rule
$\mu_p(x,y)=U_p{}^{-1}\mu(U_p x,U_p y)$ $\forall \; x, y\in V$.
If for any $x, y\in V$ there exists the limit
\[
\lim\limits_{p \to\infty}\mu_p(x,y)=
\lim\limits_{p \to\infty}U_p{}^{-1}\mu(U_p x,U_p y)=:\mu_0(x,y)
\]
then $\mu_0$ is a well-defined Lie bracket.
The Lie algebra $\mathfrak g_0=(V,\mu_0)$ is called a \emph{sequential contraction} of the Lie algebra~${\mathfrak g}$.

Any continuous contraction from $\mathfrak g$ to $\mathfrak g_0$ gives an infinite family of matrix sequences resulting in 
sequential contractions from $\mathfrak g$ to $\mathfrak g_0$.
More precisely, if $U_\varepsilon$ is the matrix of the continuous contraction and
the sequence $\{\varepsilon_p,\, p\in\mathbb N\}$ satisfies the conditions $\varepsilon_p\in(0,1]$, $\varepsilon_p\to+0$, $p\to\infty$,
then $\{U_{\varepsilon_p},\, p\in\mathbb N\}$ is a matrix sequence generating a sequential contraction from $\mathfrak g$ to $\mathfrak g_0$.

Conversely, if a sequence $\{U_p,\,p\in\mathbb N\}\subset{\rm GL}(V)$ defines a sequential contraction from $\mathfrak g$ to $\mathfrak g_0$ 
(and the sign of $\det U_p$ is the same for all $p\in\mathbb N$ if $\mathbb F=\mathbb R$)
then there exists a one-parametric continuous contraction from $\mathfrak g$ to $\mathfrak g_0$ with 
the associated continuous function $\tilde U\colon (0,1]\to {\rm GL}(V)$ such that $\tilde U_{1/p}=U_p$ for any $p\in\mathbb N$. 
The simple proof of this fact involves logarithms and exponents of matrices and, additionally,  
the polar decomposition in the real case. 

Definitions of special types of contractions, statements on properties and their proofs in the case of sequential contractions 
can be easily obtained via reformulation of those for the case of continuous contractions. 
It is enough to replace continuous parametrization by discrete one.
The replacement is invertible.

\section{Equivalence of diagonal contractions\\ to generalized IW-contractions}

There exists a class of contractions, which is wider than the class of generalized IW-contractions
and, at the same time, any contraction from this class is equivalent to a generalized IW-contraction involving only integer parameter powers.
Similar to generalized IW-contractions, this class is singled out by restrictions on contraction matrices
instead of restrictions on algebra structure.

\begin{definition}\label{DefOfDiogonalContractions}
The contraction $\mathfrak g\to\mathfrak g_0$ (over $\mathbb F=\mathbb C$ or $\mathbb R$) is called \emph{diagonal} if
its matrix $U_\varepsilon$ can be represented in the form $U_\varepsilon=AW_\varepsilon P$,
where $A$ and $P$ are constant nonsingular matrices
and $W_\varepsilon=\diag(f_1(\varepsilon),\dots,f_n(\varepsilon))$ for some
continuous functions $f_i\colon (0,1]\to\mathbb F\backslash\{0\}$.
\end{definition}

\begin{theorem}\label{TheoremOnDiagonalContractions}
Any diagonal contraction is equivalent to a generalized IW-contraction with an integer signature.
\end{theorem}

\begin{proof}
Let the contraction $\mathfrak g\to\mathfrak g_0$ have the matrix $U_\varepsilon$ of the form from Definition~\ref{DefOfDiogonalContractions}.
Due to possibility of changing bases in the initial and contracted algebras, we can set $A$ and $P$ equal to the unit matrix.
If $U_\varepsilon=W_\varepsilon$, the structure constants of the contracted algebra $\mathfrak g_0$ are calculated by the formula
\[
c_{0,ij}^k=\lim_{\varepsilon\to +0}c_{ij}^k\,\frac{f_if_j}{f_k}
\]
with no summation with respect to the repeated indices.
Therefore, the condition
\[
\exists \lim_{\varepsilon\to +0}\frac{f_if_j}{f_k}=:F_{ij}^k\in \mathbb F
\quad\mbox{if}\quad c_{ij}^k\ne 0
\]
are necessary and sufficient
for the existence of the well-defined diagonal contraction with the contraction matrix $U_\varepsilon$,
and $c_{0,ij}^k=c_{ij}^kF_{ij}^k$ if $F_{ij}^k$ is defined and belongs to $\mathbb F\backslash\{0\}$
and $c_{0,ij}^k=0$ otherwise.

Introducing the notations
\[
\mathcal E=\{(i,j,k)\mid i<j,\ c_{ij}^k\ne 0,\ F_{ij}^k\ne0\},\quad
\mathcal N=\{(i,j,k)\mid i<j,\ c_{ij}^k\ne 0,\ F_{ij}^k=0\},
\]
we associate the set of the limits $F_{ij}^k$, $(i,j,k)\in\mathcal E\cup\mathcal N$,
with two systems%
\footnote{Here we also can assume that $\mathcal N\ne\varnothing$ since otherwise the contraction $\mathfrak g\to\mathfrak g_0$ is improper 
and therefore, equivalent to a generalized IW-contraction with the zero signature.}% 
:

1) the system $\Sigma$ of the equations $y_iy_j/y_k=F_{ij}^k$, $(i,j,k)\in\mathcal E$,
for $y=(y_1,\dots,y_n)\in(\mathbb F\backslash\{0\})^n$ and

2) the mixed system $S$ of the equations $x_i+x_j-x_k=0$, $(i,j,k)\in\mathcal E$, and
the inequalities $x_i+x_j-x_k>0$, $(i,j,k)\in\mathcal N$, 
for $x=(x_1,\dots,x_n)\in\mathbb R^n$.

\noindent

We will prove using the rule of contraries that the existence of the nonzero limits $F_{ij}^k$ for $(i,j,k)\in\mathcal E$
(resp. these limits and, additionally, the zero limits for $(i,j,k)\in\mathcal N$)
implies the compatibility of the system~$\Sigma$ (resp. $S$).
The principal observation is that the solving operations with equations and inequalities of the systems
are equivalent to similar operations with the limits.

Suppose that the system~$\Sigma$ is incompatible.
Then $\mathcal E\ne\varnothing$. 
We use the multiplicative version of the Gaussian elimination, involving only integer powers. 
It is reduced to iterating the following procedure. 
Let $\Sigma_\nu$ denote the set $\{Y_s=G_s,\ s=1,\dots,\sigma\}$
of consequences of~$\Sigma$ obtained after the $\nu$th iteration, $\Sigma_0:=\Sigma$. 
Here $Y_s$ (resp.\ $G_s$) are products of integer powers of a selection of $y$'s (resp.\ $F$'s), 
the numbers $|\Sigma_\nu|$ and $|\Sigma|$ of equations of the systems $\Sigma_\nu$ and $\Sigma$ coincides, $|\Sigma_\nu|=|\Sigma|=:\sigma$.
We choose $i$ such that $y_i$ is in the system $\Sigma_\nu$ and 
denote by $\beta_s$ the degree of $Y_s$ with respect to~$y_i$.  
Let $\bar\beta=\mathop{\rm gcd}(\beta_1,\dots,\beta_\sigma)$ be the greatest common divisor of $\beta_1$, \dots, $\beta_\sigma$. 
According the generalized B\'ezout's identity, we have the representation $\bar\beta=\delta_s\beta_s$ with integer coefficients $\delta_s$. 
Consider the consequence $\bar Y=\bar G$ of $\Sigma_\nu$, where 
$\bar Y=Y_1^{\delta_1}\cdots Y_\sigma^{\delta_\sigma}$ and $\bar G=G_1^{\delta_1}\cdots G_\sigma^{\delta_\sigma}$. 
The degree of $\bar Y$ with respect to $y_i$ equals $\bar\beta$. 
The equations $Y_s\bar Y^{-\beta_s/\bar\beta}=G_s\bar G^{-\beta_s/\bar\beta}$, $s=1,\dots,\sigma$, form the system $\Sigma_{\nu+1}$. 
By the construction, the number of unknowns in $\Sigma_{\nu+1}$ is less by 1 than that in~$\Sigma_\nu$.
The incompatibility of the system~$\Sigma_\nu$ is equivalent to the incompatibility of the system~$\Sigma_{\nu+1}$.
In view of the incompatibility of $\Sigma$, after iterations we have to obtain
an inconsistent consequence of the form
\[
1=\prod_{(i,j,k)\in\mathcal E}(F_{ij}^k)^{m_{ij}^k}
\qquad\Biggl(\mbox{or}\quad 
Y^2=\prod_{(i,j,k)\in\mathcal E}(F_{ij}^k)^{m_{ij}^k}\Biggl),
\]
where the right-hand side does not equal the unity (resp.\ is negative), $m_{ij}^k\in\mathbb Z$, $(i,j,k)\in\mathcal E$, 
and $Y$ is a product of integer powers of $y$'s. 
Inconsistent consequences of the second form are related only to the case of real numbers. 
The same combination of operations is well defined for the limits with $(i,j,k)\in\mathcal E$
and, applied to them, results in the same (resp.\ similar) inconsistent equality for limits
that contradicts the existence of the diagonal contraction.

Suppose that the system $S$ is incompatible.
The subsystem of the equations $x_i+x_j-x_k=0$, $(i,j,k)\in\mathcal E$,
should have solutions. (At least, it has a zero solution.)
Applying the Gaussian elimination over $\mathbb Z$ to this subsystem,
we present it in the form
\[
a_ix_i=\sum_{j\in\bar I}b_i^jx_j, \quad i\in I,
\]
where $I\subset\{1,\dots,n\}$, $\bar I=\{1,\dots,n\}\backslash I$,
$a_i\in\mathbb N$, $b_i^j\in\mathbb Z$, $i\in I$ and $j\in\bar I$.
We eliminate $x_i$, $i\in I$, from the inequalities $x_{i'}+x_{j'}-x_{k'}>0$, $(i',j',k')\in\mathcal N$,
multiplying them by appropriate products of some of $a_i$, $i\in I$.
As a result, we obtain a system $S'$ of strict homogenous linear inequalities for $x_j$, $j\in\bar I$,
with integer coefficients.
Since the system $S$ is incompatible, there exists an identically vanishing linear combination with natural coefficients,
composed of left-hand sides of inequalities from $S'$.%
\footnote{This statement is a modification of the well-known Voronoy theorem~\cite{Voronoy1908} (see also \cite[p.~10]{Chernikov1968})
to the case of homogenous strict linear inequalities with integer (or rational) coefficients.}
This combination gives the inconsistent inequality $0>0$.

The above chain of additive operations with equations and inequalities of~$S$ is associated with
a chain of well-defined multiplicative operations with the limits $F_{ij}^k$, $(i,j,k)\in\mathcal E\cup\mathcal N$.
Under this association adding, subtracting and multiplying by integers are replaced by
multiplying, dividing and raising to the corresponding powers, respectively.
Only multiplying and raising to natural powers can be applied to
the zero limits $F_{ij}^k$, $(i,j,k)\in\mathcal N$, that agrees with restrictions on operations with inequalities.
The chain of operations with the limits leads to the inconsistent equality $1=0$
that contradicts the existence of the contraction.

Let $(\gamma_1,\dots,\gamma_n)$ and $(\alpha_1,\dots,\alpha_n)$ be solutions of the systems~$\Sigma$ and~$S$, respectively.
It is obvious that $\gamma_1\cdots\gamma_n\ne0$.
The system~$S'$ possesses rational solutions since the solution set of~$S'$ is open and nonempty.
This implies that the system~$S$ has rational solutions and hence has integer solutions,
i.e., we can assume then $\alpha_1,\dots,\alpha_n\in\mathbb Z$.
Then the matrix $\tilde U_\varepsilon=\tilde A\tilde W_\varepsilon P$,
where $\tilde A=A\diag(\gamma_1,\dots,\gamma_n)$ and
$\tilde W_\varepsilon=\diag(\varepsilon^{\alpha_1},\dots,\varepsilon^{\alpha_n})$,
generates a well-defined generalized IW-contraction from $\mathfrak g$ to $\mathfrak g_0$ with integer exponents.
\end{proof}

In other words, Theorem~\ref{TheoremOnDiagonalContractions} says that
generalized IW-contractions are universal in the class of diagonal contractions.

\begin{note}\looseness=1
Under constructing a generalized IW-contraction equivalent to a diagonal contraction in the way described in the proof,
the constant matrix factors of the associated contraction matrix are in fact known 
and coincide, up to a multiplier, with the ones of the diagonal contraction. 
The multiplier is a constant diagonal matrix whose diagonal entries give a solution of the system~$\Sigma$.
Only solving the system~$S$ of linear equations and inequalities with respect to parameter exponents is of a significant value.
We can choose an integer solution of~$S$ which is optimal in some sense, e.g.,
the maximum of the absolute values of whose components is minimal.
In general, the chosen solution may be non-optimal, in the same sense, in the entire set of integer signatures
of generalized IW-contractions from  $\mathfrak g$ to $\mathfrak g_0$.
To choose a totally optimal signature, for each tested tuple of exponents from a number of preliminary selected ones
we have either to find a solution of a cumbersome nonlinear system of algebraic equations
on coefficients of the constant matrix factors or to prove incompatibility of this system.
This is much more complicated problem than that discussed in Theorem~\ref{TheoremOnDiagonalContractions}.
\end{note}

\begin{corollary}
Any diagonal contraction whose matrix possesses a finite limit at $\varepsilon\to +0$
is equivalent to a generalized IW-contraction with nonnegative integer exponents.
\end{corollary}

\begin{proof}
Since the functions $f_i$ possess finite limits at $\varepsilon\to +0$, 
we can attach the additional equations $x_i>0$ to the system~$S$ and 
prove using the same way as in the proof of Theorem~\ref{TheoremOnDiagonalContractions} that the extended system has integer solutions.  
\end{proof}

\begin{note}
Other additional restrictions on exponents of generalized IW-contractions which are equivalent to diagonal contractions 
with certain properties can be set in a similar way.
In particular, it obviously follows from the proof of Theorem~\ref{TheoremOnDiagonalContractions} that 
for any fixed $j$ the $j$th exponent can be chosen nonnegative (negative) 
if there exists a finite (infinite) limit of $f_j$ at $\varepsilon\to +0$.
\end{note}

\begin{note}
The same theorem is true for diagonal sequential contractions, 
and its proof completely coincides with the proof of Theorem~\ref{TheoremOnDiagonalContractions}. 
\end{note}

\begin{note}
Theorem~\ref{TheoremOnDiagonalContractions} is obviously extended to the class of contractions wider than the class of diagonal contractions.
In particular, it implies that any contraction $\mathfrak g\to\mathfrak g_0$
whose matrix can be represented in the form $U_\varepsilon=\hat U_\varepsilon AW_\varepsilon \check U_\varepsilon$
is equivalent to a generalized IW-contraction with an integer signature.
Here 
$\hat U_\varepsilon$ is an automorphism matrix of the algebra $\mathfrak g$, 
$\check U_\varepsilon$ is a nonsingular matrix with the well-defined componentwise limit $\lim_{\varepsilon\to+0}\check U_\varepsilon=:P$, 
both the matrices $\hat U_\varepsilon$ and $\check U_\varepsilon$ are continuously parameterized by $\varepsilon\in(0,1]$,
$A$ and $P$ are constant nonsingular matrices 
and $W_\varepsilon=\diag(f_1(\varepsilon),\dots,f_n(\varepsilon))$ for some
continuous functions $f_i\colon (0,1]\to\mathbb F\backslash\{0\}$.
The problem on the widest class of parameterized matrices which are associated with contractions equivalent to generalized IW-contractions
is still open.
\end{note}

%\noprint

\subsection*{Acknowledgements}

The authors are grateful to Dietrich Burde, Maryna Nesterenko and Evelyn Weimar-Woods
for productive and helpful discussions.
The research of ROP was supported by the Austrian Science Fund (FWF), project P20632.

\end{document}